\documentclass[12 pt]{article}

\pdfoutput=1

\usepackage{fullpage}
\usepackage{amscd}
\usepackage{amsmath}
\usepackage{amssymb}
\usepackage{amsthm}
\usepackage{graphicx}
\usepackage{latexsym}
\usepackage{verbatim}
\usepackage{tikz}
\usetikzlibrary{backgrounds,shapes.geometric,fit,matrix,arrows}

\usepackage{fancyhdr}

\newtheorem{theorem}{Theorem}[section]

\newtheorem{corollary}[theorem]{Corollary}

\newtheorem{proposition}[theorem]{Proposition}

\theoremstyle{definition}

\theoremstyle{definition}

\theoremstyle{remark}
\newtheorem{example}[theorem]{Example}

\theoremstyle{definition} 
\newtheorem{definition}[theorem]{Definition}

\newcounter{tempthm}
\newcounter{tempsec}

\newcommand{\savecounter}[1]{\newcounter{thmcounter#1}
\setcounter{thmcounter#1}{\value{theorem}}
\newcounter{seccounter#1}
\setcounter{seccounter#1}{\value{section}}}

\newcommand{\usesavedcounter}[1]{\setcounter{tempthm}{\value{theorem}}
\setcounter{theorem}{\value{thmcounter#1}}
\setcounter{tempsec}{\value{section}}
\setcounter{section}{\value{seccounter#1}}}

\newcommand{\restorecounter}{\setcounter{theorem}{\value{tempthm}}
\setcounter{section}{\value{tempsec}}}

\newcommand{\Qm}{\mathbb{Q}}

\newcommand{\Nm}{\mathbb{N}}

\newcommand{\Rm}{\mathbb{R}}

\newcommand{\shortv}[1]{}

\DeclareMathOperator{\rank}{rank}
\DeclareMathOperator{\conv}{conv}

\DeclareMathOperator{\prob}{\mathbb{P}}
\DeclareMathOperator{\expect}{\mathbb{E}}

\DeclareMathOperator{\ce}{CE}
\DeclareMathOperator{\nash}{NE}
\DeclareMathOperator{\xe}{XE}
\DeclareMathOperator{\DNN}{DNN}

\DeclareMathOperator{\sym}{Sym}

\newcommand{\ilim}{\mathop{\varprojlim}\limits}

\newcommand{\xesym}{\xe_{\sym}}

\newcommand{\nesym}{\nash_{\sym}}
\newcommand{\cesym}{\ce_{\sym}}

\newcommand{\dsym}{\Delta_{\sym}}
\newcommand{\dsympi}{\dsym^\Pi}

\newcommand{\npowk}[2]{#2^{(#1)}}
\newcommand{\npow}[1]{\npowk{N}{#1}}

\providecommand{\abs}[1]{\lvert#1\rvert}

\newcommand{\eqdef}{:=}

\title{\LARGE \bf
Exchangeable Equilibria, Part I: \\ Symmetric Bimatrix Games
}

\author{Noah D.\ Stein,  Asuman Ozdaglar, and Pablo A.\ Parrilo
\thanks{The first author is at Analog Devices $\vert$ Lyric Labs, reachable at {\tt\small noah.stein@analog.com}.  The remaining authors are at the Laboratory for Information and Decision Systems of the Massachusetts Institute of Technology, reachable at {\tt\small asuman@mit.edu} and  {\tt\small parrilo@mit.edu}.}
      }

\begin{document}

\maketitle




\begin{abstract}
We introduce the notion of exchangeable equilibria of a symmetric bimatrix game, defined as those correlated equilibria in which players' strategy choices are conditionally independently and identically distributed given some hidden variable.  We give several game-theoretic interpretations and a version of the ``revelation principle''.  Geometrically, the set of exchangeable equilibria is convex and lies between the symmetric Nash equilibria and the symmetric correlated equilibria.  Exchangeable equilibria can achieve higher expected utility than symmetric Nash equilibria.

\end{abstract}

\section{Introduction}

The contribution of this paper is to introduce exchangeable equilibria, a solution concept for symmetric games intermediate between symmetric Nash and correlated equilibria.  For simplicity we restrict to symmetric bimatrix (two-player) games throughout.  The more technical multiplayer case will be treated in Part~II of this paper, with computational issues covered in future work.\footnote{For more details on these topics, see the first author's doctoral thesis \cite{s:doctoralthesis}.}

\subsection{Motivation}

A basic tenet of decision theory is that preferences are specific to each individual, and so interpersonal comparison of utilities is meaningless.  In practice this belief is often dropped in favor of other simplifying assumptions which can be justified within the setting at hand.  One important example of this phenomenon is zero-sum bimatrix games.  The zero-sum condition says not only that the players' preferences over outcomes are opposite, but also that their preferences over lotteries over outcomes are opposite.  This is a very strong condition indeed, and while it may never be true exactly, it is often a reasonable approximation to model situations in which there seems to be no opportunity for cooperation.

Similarly, symmetric bimatrix games are an idealization in which we take the roles and preferences of the players to be identical.  When such an approximation is reasonable, we can also investigate the consequences of a stronger assumption: that the players are independent instances of an identical decision-making agent, or clones, as it were.  Doing so, we should expect both players would make the same choice when confronted with the same situation.  They interpret their observations of the environment in the same way.  If the setup is such that we have reason to believe that the players will choose their strategies statistically independently, then our solution concept of choice should be symmetric Nash equilibrium.  Here we explicitly use the assumption that the players are identical; in particular they have not been provided with a way to break the symmetry so as to choose an asymmetric Nash equilibrium.\footnote{General bimatrix games can be symmetrized by considering two copies of the game played simultaneously with each player taking opposite roles in the two copies.  In this way results from the symmetric theory can be applied to asymmetric games.  This idea of players ``putting themselves in each other's shoes'' incorporates a notion of fairness into the analysis of asymmetric games.  Rawls explores the philosophical implications of taking this idea as the basis for a theory of justice \cite{r:tj}, though in a way which is at odds with some fundamental game-theoretic principles \cite{h:cmpsbmcjrt,b:sc1hr}.}

This raises a natural question: what is the appropriate solution concept to use if statistical independence of actions is not a reasonable assumption and some correlation may be expected?  For a general game the answer would be correlated equilibrium\footnote{See Aumann's argument in \cite{a:ceebr}.  The key assumption forming the basis of that argument is that the players have common prior beliefs.  While this is debatable in some contexts, it is certainly true in the symmetric setting considered here.  Put another way, the philosophical or epistemic argument for correlated equilibrium is even stronger in symmetric games than in general games.}, so the obvious answer in the symmetric setting is symmetric correlated equilibrium.  That is to say, we could expect any correlated equilibrium which is invariant (as a probability distribution) under interchange of the players.

This obvious answer is correct insofar as the symmetry of the situation rules out any correlated equilibria which are not symmetric.  It is our goal to argue that symmetry can also rule out some of the symmetric correlated equilibria!

To see this, we begin with the tautological statement that each player bases his action on the state of the world.  Perhaps he plays the same action at all states, perhaps not.  Insofar as he does not have full knowledge of the state of the world his action must be based on some noisy observations of the world.

Since the players are identical they both make the same measurements and interpret them in the same way, but the outcomes of these measurements may differ.  In this way and only this way may the players' actions differ.  In particular our assumption that the players are identical means that the players cannot correlate their actions based on knowledge of each other's types, or past histories of actions, or any player-specific information.  Only correlation based on the external world is possible.

For example, the players may wish to base their actions on tomorrow's weather.  This is an aspect of the underlying state of the world, albeit one which may not be known precisely to the players (it is a \emph{hidden variable}), so they can only condition their actions on their estimates of the weather.  Perfect correlation may occur if these estimates are obtained from a common weather report.  Weaker correlation may arise if both players make their own forecasts in terms of independent measurements of atmospheric data.

We define an \emph{exchangeable equilibrium} of a symmetric bimatrix game to be a correlated equilibrium distribution in which the players' actions are independent and identically distributed (i.i.d.) conditioned on some hidden variable (Definition~\ref{def:xecp}).  These are exactly the outcomes which could conceivably arise when rational players make independent environmental measurements, then interpret and act on this information in the same way.  This is a version of the ``revelation principle'' \cite{f:cecg} for exchangeable equilibria (compare Proposition~\ref{prop:symcorrintext} with Proposition~\ref{prop:corrintext}).  Exchangeable equilibria can achieve higher payoffs and social welfare than symmetric Nash equilibria (Example~\ref{ex:payoffsep}).

Over the course of the paper we motivate this definition from several viewpoints: game-theoretic, probabilistic, algebraic / geometric, and (while leaving a detailed discussion for future work) computational.  We begin by unpacking the definition.

Any probability distribution of two players' actions $X_1, X_2$ can be written as a square matrix $A$, with $A_{ij} \eqdef \prob(X_1 = i\text{ and }X_2 = j)$.  If the actions were i.i.d.\ their joint distribution could be written in matrix form as an outer product $xx^T$ for some probability vector $x$.  When the actions are merely conditionally i.i.d.\ the distribution is in fact a convex combination\footnote{The end result is the same whether we consider finite sums or infinite ``convex combinations'' weighted by probability distributions; see Section~\ref{sec:probdistexch}.} $\sum_i\lambda_i x_ix_i^T$ of terms of this form ($\lambda_i\geq 0$, $\sum_i \lambda_i =1$).

We might say that any matrix of probabilities which is not symmetric explicitly breaks symmetry, whereas a symmetric matrix which is not conditionally i.i.d.\ implicitly breaks symmetry.  We focus on the exchangeable equilibria, those correlated equilibria which do not break symmetry at all.

To illustrate this symmetry breaking, consider an $m\times m$ matrix $A$ (symmetric or not) of probabilities specifying the joint distribution of two random variables and let $e_i$ denote the $i^{\text{th}}$ unit column vector.  We can write $A = \sum_{i,j} A_{ij}e_ie_j^T$.  The matrix $e_ie_j^T$ is an independent distribution corresponding to two constants.  Therefore we can always view a sample taken according to a  probability matrix $A$ as the result of two ``measurements'' which are conditionally independent given some hidden variable.

Here we have taken the hidden variable to be the result of the two measurements.  Symmetry is broken in the sense that players' actions are not identically distributed given the hidden variable: if the value is $(i,j)$ the first player deterministically plays $i$ and the second plays $j$.  The players do not map their measurements of the hidden variable to actions in the same way.  Looked at differently, they are not measuring the same aspects of the hidden variable to determine their choices of action.

In general there are many such decompositions $A = \sum_i \lambda_i x_i y_i^T$.  We only wish to consider matrices $A$ which admit a decomposition $A = \sum_i \lambda_i x_i x_i^T$ respecting the symmetry of the problem.  Clearly lack of symmetry of $A$ is an obstruction to the existence of such a decomposition, but it is not the only obstruction: see Section~\ref{sec:preview} for an example.

%

Geometrically, the set of exchangeable equilibria arises naturally when considering convex relaxations of the set of Nash equilibria.  Let $\cesym$ denote the set of symmetric correlated equilibria of a symmetric game, and $X$ denote the set of symmetric, rank $1$, elementwise nonnegative matrices, i.e., matrices of the form $xx^T$ for column vectors $x\geq 0$.  Then we have
\begin{equation*}
\stackrel{\text{Nash}}{\cesym\cap X}\ \subseteq\ \stackrel{\text{convex hull of Nash}}{\conv(\cesym\cap X)}\ \subseteq\ \stackrel{\text{exchangeable}}{\cesym\cap\conv(X)}\ \subseteq\ \stackrel{\text{correlated}}{\cesym},
\end{equation*}
where each type of symmetric equilibrium is defined by the set written below it.  


We will discuss three further interpretations of exchangeable equilibria in this paper, in addition to the above motivations via hidden variables or geometry.  First, we show that these are the equilibria which can arise as a Bayesian observer's predictions when players are chosen from a large pool of agents whom the observer views as interchangeable.  Second, we view them as symmetric correlated equilibria of games with many players which can be broken up into identical symmetric pairwise interactions between these players.  In particular the exchangeable equilibria are those symmetric correlated equilibria which do not depend on the exact number of players.  Third, the exchangeable equilibria are those which can be implemented with a particular kind of correlation scheme.

\subsection{Preview}
\label{sec:preview}

\begin{table}[tbp]
\begin{center}
\begin{tabular}{c|c|c|} $(u_\text{row},u_\text{col})$ & Wimpy & Macho \\ \hline Wimpy & $(4,4)$ & $(1,5)$ \\  \hline Macho & $(5,1)$ & $(0,0)$ \\ \hline
\end{tabular}
\end{center}
\caption{The game of Chicken.  Two players ride their bikes towards each other.  At the last moment, each can choose to either be Wimpy and veer off to the side, or be Macho and remain on the collision course.  One player chooses a row, the other a column, and the corresponding cell of the table shows their respective utilities.}
\label{tab:chicken}
\end{table}

We illustrate the theory of exchangeable equilibria using the game Chicken (Table~\ref{tab:chicken}).  This game has two asymmetric pure Nash equilibria and one symmetric mixed Nash equilibrium. We write the corresponding joint distributions as matrices.  Each entry is the probability of choosing the corresponding strategy profile.
\begin{equation*}
\nash = \left\{\begin{bmatrix} 0 & 0 \\ 1 & 0\end{bmatrix},\begin{bmatrix} 0 & 1 \\ 0 & 0 \end{bmatrix}, \begin{bmatrix}1/4 & 1/4 \\ 1/4 & 1/4\end{bmatrix}\right\}.
\end{equation*}
In the same notation a probability matrix $\left[\begin{smallmatrix} p & q \\ q & r\end{smallmatrix}\right]$ with $p,q,r\geq 0$ and $p+2q+r=1$ is a symmetric correlated equilibrium if and only if $q\geq p,r$.  These inequalities define the polytope
\begin{equation}
\label{eq:chickencesym}
\cesym = \conv\left\{\begin{bmatrix} 0 & 1/2 \\ 1/2 & 0\end{bmatrix}, \begin{bmatrix} 1/3 & 1/3 \\ 1/3 & 0\end{bmatrix}, \begin{bmatrix} 0 & 1/3 \\ 1/3 & 1/3\end{bmatrix}, \begin{bmatrix} 1/4 & 1/4 \\ 1/4 & 1/4\end{bmatrix}\right\}.
\end{equation}

Consider the symmetric correlated equilibrium
\begin{equation*}
B\eqdef\begin{bmatrix}0 & 1/2 \\ 1/2 & 0\end{bmatrix}.
\end{equation*}
Is it an exchangeable equilibrium?  That is to ask, is it conditionally i.i.d.?  Equivalently, is it of the form $\sum \lambda_i x_i x_i^T$?

Any matrix of this form would be positive semidefinite, but $\det B = -\frac{1}{4}<0$, so $B$ is not positive semidefinite. It is a symmetric correlated equilibrium which is not exchangeable.  In any game positive semidefiniteness is necessary for a correlated equilibrium to be exchangeable.  Known results (Theorem~\ref{thm:cpvsdnn}) show it is sufficient in games with at most four strategies.

The same test rules out the second and third extreme points listed in \eqref{eq:chickencesym} as not exchangeable.  In fact for Chicken the condition that the determinant be nonnegative ($q^2\leq pr$) and the symmetric correlated equilibrium conditions have only one common solution:
\begin{equation*}
\left\{\begin{bmatrix}1/4 & 1/4 \\ 1/4 & 1/4\end{bmatrix}\right\} = \nesym = \xesym \subsetneq \cesym.
\end{equation*}
This example illustrates several general properties of the set of exchangeable equilibria that we will establish in this paper: $\conv(\nesym)\subseteq\xesym$ (Proposition~\ref{prop:nashsubsetxe}) with equality in games with two strategies (Theorem~\ref{thm:2by2}) and $\xesym\subsetneq\cesym$ (Proposition~\ref{prop:nashsubsetxe}) with strict containment whenever asymmetric Nash equilibria exist (Theorem~\ref{thm:asymnashstrict}).

\subsection{Previous work}
As we will see, the idea of exchangeable equilibrium arises naturally when one combines the notion of exchangeable random variables with the correlated equilibrium solution concept.  Both of these are large areas of study in their own right and we make no attempt at an exhaustive survey of the literature of one or the other.  However, several papers either have led toward the idea of exchangeable equilibrium in some way or seem to be particularly related to what is accomplished by exchangeable equilibria and we review those here.

One such example is Hart and Schmeidler's minimax proof of the existence of correlated equilibria \cite{hs:ece} or the similar proof by Nau and McCardle \cite{nm:cbng}.  These arguments share the somewhat mysterious quality that at an intermediate stage they work with product distributions which are not equilibria, but in the end, just when a correlated equilibrium is found, the product structure vanishes.  The result is a seemingly arbitrary correlated equilibrium with no apparent extra structure; in particular the equilibrium need not have the product structure which would make it a more desirable (in the sense of strength of the existence proof) Nash equilibrium.  Considering what extra structure might be obtained from such an argument, in particular if we require the game to be symmetric and try to make the argument respect the symmetry as much as possible, is one route to our notion of exchangeable equilibrium.  Two existence proofs along these lines will be presented in Part~II of this paper.

Another example is Brandenburger's survey of epistemic game theory \cite{b:keg}.  He mentions that the standard assumption of probabilistic independence of different players' strategies in noncooperative game theory is somewhat suspect and perhaps unnatural.  In a footnote he offers the concept of exchangeability as an example of a more natural way to capture our ignorance of the distinctions between random variables.  However, this is not explored further.

The notion of exchangeable equilibrium bears at least an outward resemblance to the work of Hillas, Kohlberg, and Pratt on an outside observer's assessment of the outcome of a game \cite{hkp:ceneoag}.  They consider an observer watching a given $n$-player game being played repeatedly.  Each time the game is played by a new set of players disjoint from the set of all previous players, so the observer views these interactions as exchangeable: he has no preconceived notions about how the order of these repeated plays relates to their outcomes.  In other words, his prior distribution over actions is not affected by simultaneously swapping all $n$ players in copy $j$ of the game with the corresponding $n$ players in copy $k$ for any $j,k$.  The players do not have access to the history of play, though the observer does.

Hillas, Kohlberg, and Pratt characterize correlated equilibria of the original game by considering whether at each play of the game the observer can offer advice based on his past observations enabling some player to increase his expected payoff.  They argue that the players should have a better understanding of the game than the observer, and prove that the limiting distribution of play is a correlated equilibrium if and only if the observer cannot offer such helpful advice.

If the observer also does not have any preconceived notions about the connections between the players in each instance of the game (i.e., his prior distribution on actions is not affected by swapping the copy of player $i$ in instance $j$ of the game with the copy in instance $k$ for any $i,j,k$), then Hillas, Kohlberg, and Pratt show that the long-run distribution of play will be a Nash equilibrium.  The paper \cite{hkp:ceneoag} does not consider solution concepts between Nash and correlated equilibrium.  However, its use of exchangeability suggests an ideological kinship with the present work. It seems likely that the concept of exchangeable equilibrium could be reinterpreted in that setup, though we do not do so here.

One interpretation of exchangeable equilibria is that they are those correlated equilibria which extend to symmetric correlated equilibria of games with an arbitrarily large number of identical interactions (Section~\ref{sec:manyplayerinterp}).  That is to say, they are robust to the number of players: the players could not profitably deviate even if they knew the number of players.  This is in contrast to the paper of Myerson on games with many players in which the number of players is modeled probabilistically \cite{m:pupg}.  That work is in a Bayesian game setting and so not directly comparable to ours, but the distinction between which equilibria are sensitive to a given probabilistic model and which equilibria are not is worth making.  In a given situation, one assumption or the other may be more natural.

Another interpretation of exchangeable equilibria is that they are those correlated equilibria in which the correlating device takes a specific form: the players choose strategies i.i.d.\ conditioned on some hidden variable (Section~\ref{sec:hiddenvarinterp}).  In general it is an interesting problem to study which equilibria can arise when different assumptions are made about the correlating device.  Of course the most well-studied class of correlating device makes players choose their strategies independently and corresponds to the mixed Nash equilibria.  But there is also Sorin's notion of ``distribution equilibria,'' which are correlated equilibria in which each player gets the same payoff conditional on all outcomes and which are motivated by arguments from evolutionary game theory \cite{s:dede}.  Another good example of this type is Du's classification of the correlated equilibria which can arise when the players correlate on their hierarchy of beliefs about the play of the game \cite{d:cevhb}.



\subsection{Outline}
We begin with Section~\ref{sec:background}, fixing notation and reviewing background material on exchangeability, conditionally i.i.d.\ random variables, games, and equilibria.  With all this in place in Section~\ref{sec:xebasic} we introduce exchangeable equilibria and establish their fundamental properties.  The two sections which follow are Section~\ref{sec:interp} on interpretations and Section~\ref{sec:xeexamples} on examples.  These are largely independent of each other and can be read in either order.  We close with Section~\ref{sec:conclusions} on directions for future work.  This final section also outlines the content of Part~II of this paper.

\section{Background and notation}
\label{sec:background}
In this section we recall some known results to fix notation and language.  For further information on exchangeability see e.g.\ \cite{d:ffdf} and on games and equilibria see e.g.\ \cite{ft:gt}.

\subsection{Probability distributions, exchangeability, and conditionally i.i.d.\ random variables}
\label{sec:probdistexch}
For a (typically finite) set $T$ we write $\Delta(T)$ to denote the set of probability distributions\footnote{In case $T$ is infinite, such as when $T = \Delta(S)$ for a finite set $S$, we allow only the Borel probability measures with respect to an implied topology.} on $T$.  The mass of a singleton $t\in T$ under a probability distribution $\pi\in\Delta(T)$ will be denoted $\pi(t)\eqdef \pi(\{t\})$ for simplicity.  For finite $T$ we will view $\Delta(T)$ as a set of column vectors and $\Delta(T^2)\eqdef\Delta(T\times T)$ as a set of matrices.  All random variables will take values in a set $T$ with $2\leq \lvert T\rvert < \infty$ to avoid trivialities on one end and unnecessary complications on the other.  We write $T^n$ and $T^\infty$ for the products of $n$ and countably many copies of $T$, respectively.

Throughout the paper we will deal with collections of random variables which are \emph{exchangeable}, meaning they are in some sense indistinguishable.  Such is the case when the measurements they represent are unlabeled or labeled in a way which does not relate to their outcome.  We will usually assume these are given to us in some arbitrary order.  The defining property of exchangeability, then, is that the joint distribution does not depend on this order.

\begin{definition}
For finite $n$, a probability distribution in $\Delta(T^n)$ is \textbf{$n$-exchangeable} if it is invariant under permutations of the indices, or in other words, if it represents a sequence  $X_1,\ldots, X_n$ of $T$-valued random variables which is equal in distribution to $X_{\sigma(1)},\ldots,X_{\sigma(n)}$ for all permutations $\sigma$ on $\{1,\ldots,n\}$.  The set of $n$-exchangeable distributions is denoted $\dsym(T^n)$.
\end{definition}

A typical example of an $n$-exchangeable sequence arises by sampling without replacement.  If $4$ white balls and $6$ black balls are placed in a bag, shaken, and then removed one by one, the distribution of the resulting sequence of colors is $10$-exchangeable.

\begin{definition}
A probability distribution in $\Delta(T^\infty)$ representing a sequence $X_1,X_2,\ldots$ is said to be \textbf{$(\infty)$-exchangeable} if it is invariant under permutations of finitely many indices\footnote{This finiteness condition is natural in that it allows us to identify $\dsym(T^\infty)$ with the inverse limit $\ilim\dsym(T^n)$ using the Kolmogorov Consistency Theorem ($12.1.2$ in \cite{d:rap}, for example).}, or in other words if the distribution of $X_1,\ldots,X_n$ is $n$-exchangeable for all $n\in\Nm$.  The set of exchangeable distributions is denoted $\dsym(T^\infty)$.  We call a finite or infinite sequence of random variables \textbf{$n$-exchangeable} of \textbf{$(\infty)$-exchangeable}, respectively, if its distribution is.
\end{definition}

One example of exchangeability is a sequence of i.i.d.\ fair coin flips, or indeed any i.i.d.\ sequence.  However, this does not exhaust the possibilities.  Another example is the case in which $X_i = X_1$ almost surely for all $i$.  This is exchangeable regardless of the distribution on $X_1$, but is only independent in the trivial case when $X_1$ is deterministic.

The definition of exchangeability amounts to a collection of linear equations a given probability distribution may or may not satisfy.  This means the set of exchangeable distributions is convex: an arbitrary mixture of exchangeable distributions is also exchangeable.  Therefore a sequence of random variables which are i.i.d.\ conditioned on some other random variable is exchangeable.  In the case of an infinite sequence of exchangeable random variables, the converse is also true:

\begin{definition}
\label{def:condiid}
A (finite or infinite) sequence of random variables is called \textbf{conditionally i.i.d.} if it is i.i.d.\ conditioned on some\footnote{As always in probability theory, we assume the underlying sample space is rich enough to define whatever additional random variables we need.  Alternatively, we could extend the definition of conditionally i.i.d.\ to include any sequence of random variables equal in distribution to a sequence which is conditionally i.i.d.\ in the sense of Definition~\ref{def:condiid}.} random variable.
\end{definition}

\begin{theorem}[de Finetti]
\label{thm:definetti}
An infinite sequence $X_1, X_2, \ldots$ of random variables is exchangeable if and only if it conditionally i.i.d.
\end{theorem}

For a proof see e.g.\ \cite{s:tos}.  The upshot is that we can think of any exchangeable sequence of random variables as arising from a two-stage process\footnote{The classic example of an exchangeable sequence which does not obviously arise in this way is the P\'{o}lya urn model \cite{p:urn}.}.  First, a random element $\Theta\in\Delta(T)$ is secretly selected according to some distribution in $\Delta(\Delta(T))$.  Second, the sequence of random variables $X_1, X_2,\ldots$ is constructed to be i.i.d.\ with distribution $\Theta$.

The following two results summarize some facts about conditionally i.i.d.\ pairs $X_1, X_2$ and the set of all distributions thereof.  Proposition~\ref{prop:excp} shows that such pairs take a simple form and will often be invoked implicitly.

\begin{proposition}
\label{prop:cpgeo}
Consider the subset of $\Delta(T^2)$ corresponding to distributions of pairs $X_1,X_2$ which are conditionally i.i.d.\ given a parameter taking only finitely many values, i.e., those which can be written in matrix form as
\begin{equation}
\label{eq:finitepar}
\sum_{i=1}^k \lambda_i w_i w_i^T\text{ for some }k, \text{ with all }\lambda_i>0\text{ and }\sum_{i=1}^k\lambda_i=1.
\end{equation}
This set is convex, compact, and semialgebraic (definable by finitely many polynomial inequalities).
\end{proposition}

\begin{proof}
Call the set in question $Z$.  Let $X$ be the set of i.i.d.\ distributions in $\Delta(T^2)$.  Then $X$ is compact and $Z$ is the convex hull of $X$, so compact and convex\cite{bno:convex}.  The set $X$ is semialgebraic.  By Carath\'{e}odory's theorem any element of $Z$ can be written as a convex combination of a bounded number of elements of $X$, so $Z$ can be described by a first-order formula over the real-closed field $\Rm$.  Quantifier elimination then shows that $Z$ is semialgebraic \cite{bpr:arag}.
\end{proof}

\begin{proposition}
\label{prop:excp}
If a pair $X_1, X_2$ is conditionally i.i.d.\ then we may assume without loss of generality that the associated parameter takes its values in a finite set.  Equivalently, the joint distribution of $X_1,X_2$ is of the form \eqref{eq:finitepar}.  That is to say, the set described in Proposition~\ref{prop:cpgeo} is the set of \emph{all} conditionally i.i.d.\ distributions. 
\end{proposition}

\begin{proof}
Suppose the pair $X_1,X_2$ is i.i.d.\ conditioned on some parameter $\Theta$ with (not necessarily finitely-supported) distribution $\lambda\in\Delta(\Delta(T))$.  The unconditional joint distribution of $X_1$ and $X_2$ is $W \eqdef \int ww^T\, d\lambda(w)$.  By definition of the Lebesgue integral, $W$ is the limit of conditionally i.i.d.\ distributions each of whose parameters takes values in some finite set.   The set thereof is compact (Proposition~\ref{prop:cpgeo}), so it is closed and contains $W$.
\end{proof}

How can one be sure that two random variables are not conditionally i.i.d.?  In other words, given a probability distribution as a matrix, how can one certify that a decomposition of the form \eqref{eq:finitepar} does not exist?  A weak but easily applied condition is:

\begin{proposition}
\label{prop:cpzeros}
If $X_1,X_2$ are conditionally i.i.d.\ random variables taking values in the finite set $T$ and $\prob(X_1 = X_2 = t) = 0$ then $\prob(X_1 = t, X_2 = \tilde{t})=0$ for all $\tilde{t}\in T$.
\end{proposition}

\begin{proof}
Let $\Theta$ be the hidden parameter conditioned on which the $X_i$ are i.i.d.  By assumption $\Theta$ satisfies $\prob(X_1 = t \mid \Theta)=0$ almost surely, so $\prob(X_1 = t) = \int \prob(X_1 = t\mid\Theta)\,d\prob(\Theta) = 0$.
\end{proof}

Using this condition one can immediately construct symmetric, elementwise nonnegative matrices which are not the joint distribution of any conditionally i.i.d.\ pair, such as $\left[\begin{smallmatrix} 0 & 1/2 \\ 1/2 & 0 \end{smallmatrix}\right]$.  A more broadly applicable necessary condition is double nonnegativity.

\begin{definition}A matrix is \textbf{doubly nonnegative} if it is symmetric, elementwise nonnegative, and positive semidefinite.  The set of doubly nonnegative $m\times m$ matrices is denoted $\DNN_m$.
\end{definition}

The advantage of this definition is that double nonnegativity of a matrix is easy to check, and the definitions immediately yield the forward direction in:

\begin{theorem}
\label{thm:cpvsdnn}
If an $m\times m$ matrix is the joint distribution of a conditionally i.i.d.\ pair then it is doubly nonnegative and its entries sum to one.  The converse holds if and only if $m\leq 4$.
\end{theorem}

Maxfield and Minc's proof of the converse (which is not obvious from the definitions) and a counterexample in the $m=5$ case due to Hall can be found in \cite{mm:mexxa}.  Diananda proved an equivalent result about the duals of these cones \cite{d:onnfrvsawnn}.  For our purposes, the utility of this result is that it allows us to decide quickly whether $4\times 4$ and smaller matrices are joint distributions of conditionally i.i.d.\ pairs.

\subsection{Games and equilibria}
\label{sec:games}
\begin{definition}A \textbf{(finite) game}, denoted by $\Gamma$ throughout, consists of $n\geq 2$ players, each equipped with a finite set $C_i$ of two or more \textbf{(pure) strategies} and a utility function $u_i: C\to\Rm$, which assigns a real number to each \textbf{strategy profile} in $C\eqdef \prod_i C_i$.

A two-player game is also called a \textbf{bimatrix game}, because the utilities can be specified by two matrices $A,B\in\Rm^C$: $A_{ij} = u_1(i,j)$ and $B_{ij} = u_2(i,j)$.  A bimatrix game is \textbf{symmetric} if $B = A^T$, or in other words if $C_1 = C_2$ and $u_1(s_1,s_2) = u_2(s_2,s_1)$ for all $(s_1,s_2)\in C$.  For a symmetric bimatrix game we let $m = \abs{C_1}$.  When it is notationally convenient to do so, we will assume without loss of generality that $C_1=\{1,\ldots,m\}$.
\end{definition}

We make the standard assumptions that the number of players, strategy sets, and utility functions are common knowledge: everyone knows them, everyone knows that everyone knows them, and so on (see e.g.\ \cite{ft:gt}).  For symmetric bimatrix games we make the additional framing assumptions that each of a player's strategies has a distinct label, the same set of labels is used by both players, and this shared labeling is common knowledge.  In other words there is a commonly known bijection between the strategy sets of the players.  We will always take these sets to be equal and the distinguished bijection to be the identity.  Changing this bijection has a profound effect on the game as we see in Examples~\ref{ex:coord} and~\ref{ex:anticoord}.  For further discussion of symmetry and framing see \cite{ak:hsfp,c:fpfsf,s:tfp}.

\begin{definition}
A \textbf{mixed strategy} for player $i$ is a probability distribution over his pure strategy set $C_i$, and the set of mixed strategies for player $i$ is $\Delta(C_i)$.  The set of \textbf{independent distributions} or \textbf{mixed strategy profiles} of the game $\Gamma$ will be denoted\footnote{We write $\Delta^\Pi(\Gamma)$ instead of $\Delta^\Pi(C)$ because $\Gamma$ specifies the product structure on $C$.} $\Delta^\Pi(\Gamma) \eqdef \prod_i \Delta(C_i)$.  For a symmetric bimatrix game $\Gamma$, the set of symmetric product distributions will be denoted
\begin{equation*}
\dsympi(\Gamma) \eqdef \{\pi\times \pi \mid \pi\in\Delta(C_1)\} \subset\Delta^\Pi(\Gamma).
\end{equation*}
\end{definition}

For uniformity of notation we define $\Delta(\Gamma)\eqdef\Delta(C)$.  We extend each player's utility from $C$ to $\Delta(\Gamma)$ linearly, defining $u_i(\pi) = \sum_{s\in C}u_i(s)\pi(s)$ for all $\pi\in\Delta(\Gamma)$.

\begin{definition}
\label{def:nasheq}
A mixed strategy profile $\pi = (\pi_1,\ldots,\pi_n)\in\Delta^\Pi(\Gamma)$ is a \textbf{Nash equilibrium} if $u_i(\pi)\geq u_i(t_i,\pi_{-i})$ for all players $i$ and all $t_i\in C_i$.  The set of Nash equilibria is denoted $\nash(\Gamma)$.
\end{definition}

\begin{definition}
\label{def:correq}
A distribution $\pi\in\Delta(\Gamma)$ is a \textbf{(direct\footnote{Some authors use the term ``canonical'' in place of ``direct''.  We omit such qualifiers until Section~\ref{sec:hiddenvarinterp}, where we will need them to draw a distinction with a more general (``indirect'') notion of correlation.}) correlated equilibrium} if for all players $i$ and all functions $f: C_i\to C_i$
\begin{equation*}
\expect_\pi  \left[u_i(f(s_i),s_{-i}) - u_i(s)\right] \eqdef \sum_{s\in C} \left[u_i(f(s_i),s_{-i})-u_i(s)\right]\pi(s)\leq 0.
\end{equation*}
Equivalently, $\pi$ is a correlated equilibrium if for all $i$ and all $s_i,t_i\in C_i$,
\begin{equation*}
\sum_{s_{-i}\in C_{-i}} \left[u_i(t_i,s_{-i}) - u_i(s)\right]\pi(s)\leq 0.
\end{equation*}
The set of correlated equilibria is denoted $\ce(\Gamma)$.
\end{definition}

The interpretation is that $\pi\in\Delta(\Gamma)$ is a joint distribution of recommended strategies for the $n$ players.  Each player sees only his recommended strategy and $\pi$ is a correlated equilibrium when players are best off following these recommendations rather than deviating.

Viewing $\nash(\Gamma)$ and $\ce(\Gamma)$ as subsets of $\Delta(\Gamma)$, we have $\nash(\Gamma) = \ce(\Gamma)\cap\Delta^\Pi(\Gamma)$.  \textbf{Nash's Theorem} states that Nash equilibria exist, hence so do correlated equilibria \cite{nash:ncg}.  In fact in the same paper Nash also proved a version of this theorem for symmetric games; a special case of this result is that symmetric bimatrix games admit symmetric Nash equilibria as defined below.

\begin{definition}
Let $\Gamma$ be a symmetric bimatrix game.  A \textbf{symmetric Nash equilibrium}\footnote{Note that symmetric equilibrium concepts are defined as usual in terms of robustness to unilateral (and so in a sense, asymmetric) deviations, \emph{not} simultaneous identical deviations.} is a $(\pi_1,\pi_2)\in\nash(\Gamma)$ with $\pi_1 = \pi_2$.  The set of such equilibria is denoted $\nesym(\Gamma)\subseteq\dsympi(\Gamma)$.  We will at times refer to $\pi_1$ itself as a symmetric Nash equilibrium.  A \textbf{symmetric correlated equilibrium} is a $\pi\in\ce(\Gamma)$ such that $\pi(s_1,s_2) = \pi(s_2,s_1)$ for all $s\in C$, and the set of these equilibria is denoted $\cesym(\Gamma)$.
\end{definition}

In symbols, $\nesym(\Gamma) = \ce(\Gamma)\cap\dsympi(\Gamma)$ and $\cesym(\Gamma) = \ce(\Gamma)\cap\dsym(\Gamma)$.  Throughout the remainder of the paper we will be discussing symmetric bimatrix games and symmetric Nash and correlated equilibria; we will occasionally drop the word ``symmetric'' for brevity.

From the definitions we obtain the standard facts:

\begin{proposition}
\label{prop:eqtop}
For any game $\Gamma$ the set $\nash(\Gamma)$ is compact and the set $\ce(\Gamma)$ is a compact convex polytope.  For a symmetric game the same is true of $\nesym(\Gamma)$ and $\cesym(\Gamma)$.
\end{proposition}

\section{Basics of exchangeable equilibria}
\label{sec:xebasic}

We use $\Gamma$ to denote a symmetric bimatrix game throughout.

\begin{definition}
\label{def:xecp}
A \textbf{(symmetric) exchangeable equilibrium} of $\Gamma$ is a correlated equilibrium such that the players' actions are conditionally i.i.d.  The set of such equilibria is denoted
\begin{equation*}
\xesym(\Gamma)\eqdef \ce(\Gamma)\cap\conv\dsym^\Pi(\Gamma).
\end{equation*}
The text definition agrees with the equation by Proposition~\ref{prop:excp}.  In this paper we leave a notion of ``asymmetric exchangeable equilibrium'' undefined, so we are free to drop the word ``symmetric.''
\end{definition}

Here we study geometric properties of exchangeable equilibria and relationships with symmetric Nash and correlated equilibria.  Applying the equivalent characterizations of conditionally i.i.d.\ sequences presented in Section~\ref{sec:probdistexch} yields equivalent characterizations of exchangeable equilibria; we explore interpretations in Section~\ref{sec:interp}.  We compute exchangeable equilibria of examples in Section~\ref{sec:xeexamples}, which is largely independent of Section~\ref{sec:interp}.

\begin{proposition}
The set $\xesym(\Gamma)$ is compact, convex, and semialgebraic.
\end{proposition}

\begin{proof}
The set $\ce(\Gamma)$ is a compact, convex polytope (Proposition \ref{prop:eqtop}) and $\conv\dsym^\Pi(\Gamma)$ is closed, convex, and semialgebraic (Proposition~\ref{prop:cpgeo}).
\end{proof}

\begin{proposition}
\label{prop:nashsubsetxe}
The equilibrium sets are nested:
\begin{equation*}
\conv(\nesym(\Gamma))\subseteq\xesym(\Gamma)\subseteq\cesym(\Gamma).
\end{equation*}
\end{proposition}

\begin{proof}
A symmetric Nash equilibrium is an i.i.d.\ correlated equilibrium, so it is an exchangeable equilibrium.  The set of exchangeable equilibria is convex, so it contains all mixtures of these.  Exchangeable equilibria are symmetric correlated equilibria by definition.
\end{proof}

The set of symmetric correlated equilibria is always a compact convex polytope and the convex hull of the symmetric Nash equilibria is generically so, since there are generically finitely many symmetric Nash equilibria.  The set of exchangeable equilibria need not be polyhedral:  see Example~\ref{ex:exeqsep} below.  It can be shown that this set is always semialgebraic, i.e., described by finitely many polynomial inequalities (proof omitted).

\begin{theorem}
\label{thm:ratxe}
There exists a rational\footnote{Here we mean that the coordinates are rational in the algebraic sense, i.e., ratios of integers.  By definition exchangeable equilibria are rational in the game-theoretic incentive-compatibility sense.} exchangeable equilibrium: $\xesym(\Gamma)\cap\Qm^{m\times m}\neq\emptyset$.
\end{theorem}

\begin{proof}
Nash's Theorem \cite{nash:ncg} gives a symmetric Nash equilibrium $(w,w)$.  It is known that $w$ can be taken to be in $\Qm^m$.  For example we can take $w$ to be an extreme point of a certain polyhedron defined by linear inequalities with rational coefficients, by a symmetric version of the argument in \cite{m:epbg}.  Then $ww^T\in\xesym(\Gamma)\cap\Qm^{m\times m}$.
\end{proof}

The existence of exchangeable equilibria can be proven, without invoking Nash's Theorem or fixed point results, by extending the methods of \cite{hs:ece}.  We defer details to Part~II of this paper for two reasons.  First, the proofs go through in the more general context treated there.  Second, it is not clear whether such methods can prove the existence of a \emph{rational} exchangeable equilibrium, a result which we show in Part~II to be specific to the two-player case.

While symmetric games always admit symmetric Nash equilibria, asymmetric Nash equilibria can also appear (e.g.\ in Chicken and Examples~\ref{ex:anticoord} and~\ref{ex:exeqsep} below).  In games with asymmetric equilibria the set of symmetric correlated equilibria always strictly contains the exchangeable equilibria:

\begin{theorem}
\label{thm:asymnashstrict}
If $\nesym(\Gamma)\subsetneq\nash(\Gamma)$ then $\xesym(\Gamma)\subsetneq\cesym(\Gamma)$.
\end{theorem}

\begin{proof}
Let $(x,y)\in\nash(\Gamma)$ with $x\neq y$.  Both $xy^T$ and $yx^T$ are correlated equilibria, so $W\eqdef\frac{1}{2}(xy^T + yx^T)\in\cesym(\Gamma)$.  Since $x,y\in\Delta(C_1)$ are distinct, there is a $z$ with $z^Tx < 0$ and $z^T y > 0$.  This means $W$ is not positive semidefinite so not conditionally i.i.d.\ (Theorem~\ref{thm:cpvsdnn}):
\begin{equation*}
z^TWz = \frac{1}{2}(z^Txy^Tz + z^Tyx^T z) = (z^Tx)(z^Ty)<0.\qedhere
\end{equation*}
\end{proof}

In general exchangeable equilibria need not be mixtures of symmetric Nash equilibria (Example~\ref{ex:exeqsep} below), but for two-strategy symmetric bimatrix games this is always the case:

\begin{theorem}
\label{thm:2by2}
If $\Gamma$ has $m=2$ strategies per player then $\conv(\nesym(\Gamma)) = \xesym(\Gamma)$.
\end{theorem}

\begin{proof}
Let $A = \left[\begin{smallmatrix} x & y \\ z & w\end{smallmatrix}\right]$ be the payoff matrix for the row player, $b \eqdef x-z$ and $c \eqdef w - y$.  An exchangeable equilibrium $W =\left[\begin{smallmatrix} p & q \\ q & r\end{smallmatrix}\right]$ must satisfy the correlated equilibrium constraints and be conditionally i.i.d., which is the same as double nonnegativity per Theorem~\ref{thm:cpvsdnn}.  Altogether the conditions are:
\begin{equation*}
\begin{array}{rl}
pr \geq q^2 & \text{(semidefiniteness),}\\
p,q,r\geq 0 & \text{(nonnegativity),}\\
b p \geq c q & \text{(incentive constraint }\#1\text{),}\\
cr \geq b q & \text{(incentive constraint }\#2\text{), and}\\
p+2q+r = 1 & \text{(normalization).}
\end{array}
\end{equation*}

Let us show that any extreme point $W$ of $\xesym(\Gamma)$ is in $\nesym(\Gamma)$.  If $pr=q^2$ then $\rank(W)=1$, so it is a Nash equilibrium.  If $q=0$ then we can write $W = p\left[\begin{smallmatrix} 1 & 0 \\ 0 & 0\end{smallmatrix}\right] + r\left[\begin{smallmatrix} 0 & 0 \\ 0 & 1\end{smallmatrix}\right]$ as a convex combination of symmetric pure Nash equilibria.  By extremality $p = 0$ or $r= 0$.

The remaining case is that $pr>q^2>0$ so semidefiniteness and nonnegativity are not tight.  All other constraints are linear, so extremality requires tightness and linear independence of as many of these as there are variables, i.e.\ all three.  Multiplying the tight incentive constraints yields $(bc)pr = (bc) q^2$.  But $pr>q^2$, so $b = 0$ or $c=0$.  Then $bp = cq$ and $p,q>0$ give $b=c=0$.  The incentive constraints are trivial, hence linearly dependent, so no extreme points fall into this case.
\end{proof}

\section{Interpretations of exchangeable equilibria}
\label{sec:interp}
We next study various game-theoretic setups in which exchangeable equilibria arise naturally as outcomes of games.  Of course these are all related (some more obviously than others), but each has something different to add to the overall picture.  Taken together, the fact that these interpretations each give rise to the same solution concept leads us to conclude that exchangeable equilibria are natural objects to study.

\subsection{Hidden variable interpretation}
\label{sec:hiddenvarinterp}

In this section we formalize the discussion of hidden variables given in the introduction.  We do so by first fixing terminology for a more explicit notion of correlated equilibrium, then studying these in a symmetric setting.

\subsubsection{Indirect correlated equilibria}
To study hidden variables we consider an equivalent definition of correlated equilibria in which the given game is augmented by adding a correlating device which sends the players pre-play signals.  The framework is illustrated in Figure~\ref{fig:indirectce}.  The results of this subsection are known \cite{a:scrs, a:ceebr} and presented here for comparison with the Section~\ref{sec:symextcorr}.

There is some underlying (random) state of the world which is not assumed to be known to the players.  The information available to each player is an arbitrary function of this state and some private random noise; the state and all the noise signals are assumed independent.  This random noise could for example be measurement noise, or the outcome of some private coin tosses\footnote{See Proposition~\ref{prop:knownstatenash} for an important distinction between measurement noise and private coin tosses and see Example~\ref{ex:exeqsep} for an illustration thereof.} on which a player will base his action -- anything which we wish to explicitly assume the other players cannot access.

For mathematical simplicity we typically think of the state of the world, the noises, and the players' information as all taking values in some finite sets.  If this is not the case we must make some measurability restrictions.  We speak informally and avoid assigning symbols to everything involved to prevent an explosion of notation.

\begin{figure}[tbp]
\begin{center}
\begin{tikzpicture}
[roundnode/.style={circle,draw,text centered}]
\node (State) at (0,3) [roundnode, text width = 2cm] {State of the world};

\node (Add1) at (2.5,6) [roundnode] {``$+_1$''};
\node (Noise1) at (2.5,7.5) {$\text{Noise}_1$};
\node (f1) at (6.3,6) [roundnode] {$f_1$};
\node (u1) at (9.8,6) {\textcolor{white}{HIDEM}};

\node (Addn) at (2.5,0) [roundnode] {``$+_n$''};
\node (Noisen) at (2.5,1.5) {$\text{Noise}_n$};
\node (fn) at (6.3,0) [roundnode] {$f_n$};
\node (un) at (9.8,0) {\textcolor{white}{HIDEM}};

\node (Utilities) [fit = (u1) (un),draw,rounded corners,inner sep=0] {Game};

\draw [->,thick] (State) -- (Add1);
\draw [->,thick] (Noise1) -- (Add1);
\draw [->,thick] (Add1) -- (f1) node [above,midway,text width = 2.5cm,text centered] {Player $1$'s \\ information};
\draw [->,thick] (f1) -- (u1) node [above, midway, text width = 2.5cm, text centered] {Player $1$'s \\ action};
\draw [->,thick] (u1) -- (12.8,6) node [above, midway,text width = 2.5cm, text centered] {Player $1$'s \\ utility};

\draw [->,thick] (State) -- (Addn);
\draw [->,thick] (Noisen) -- (Addn);
\draw [->,thick] (Addn) -- (fn) node [above,midway,text width = 2.5cm,text centered] {Player $n$'s \\ information};
\draw [->,thick] (fn) -- (un) node [above, midway, text width = 2.5cm, text centered] {Player $n$'s \\ action};
\draw [->,thick] (un) -- (12.8,0) node [above, midway,text width = 2.5cm, text centered] {Player $n$'s \\ utility};

\foreach \y in {2.75,3.25,3.75} {
 \foreach \x in {2.5,6.3,12} {
   \fill [black] (\x,\y) circle (1.5 pt);
  }
}

\begin{pgfonlayer}{background}
\node (Hidden) [fit=(State) (Addn) (Noise1),draw=gray,fill=gray!30, rounded corners,label=above:\textcolor{gray}{Unobserved}] {};
\end{pgfonlayer}
\end{tikzpicture}
\caption[Indirect correlated equilibrium information structure]{In a correlation scheme, each player receives noisy information about the state of the world (the ``$+_i$'' indicates the state is being combined with the noise somehow, not necessarily additively) and chooses his action in the game as a function $f_i$ of this information.  If no player can improve his utility by playing a different function of his information, we call all this data (the functions and the information structure together) an indirect correlated equilibrium.}
\label{fig:indirectce}
\end{center}
\end{figure}
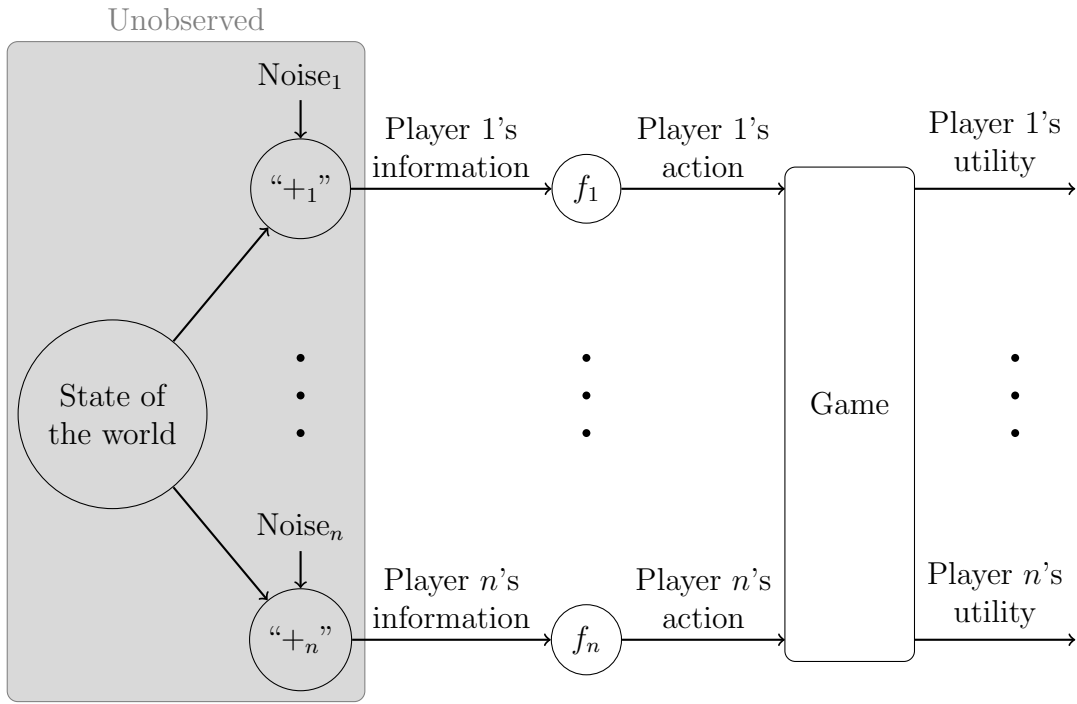

Each player chooses a function $f_i$ mapping his information to actions.  We refer to all the data together --  the state distribution, the noise distributions, the maps from these to information, and the $f_i$ -- as a \textbf{correlation scheme}.  We assume these data are known to the players: only the realizations of the random quantities are hidden.  If no player can improve his expected utility by unilaterally deviating to a different function $f'_i$, we refer to the correlation scheme as an \textbf{indirect correlated equilibrium}.  When we wish to emphasize the distinction we will refer the content of Definition~\ref{def:correq} as a \textbf{direct correlated equilibrium}.

We use the term ``indirect'' inclusively to mean ``not necessarily direct''; the proof of Proposition~\ref{prop:corrintext} shows how to view any direct correlated equilibrium as an indirect correlated equilibrium.  The term ``indirect'' does not seem to be standard, but we adopt it here to have an explicit alternative to ``direct''.

These two notions of correlation were studied by Aumann in \cite{a:scrs} and \cite{a:ceebr}, respectively, and are closely related by the ``revelation principle'' \cite{f:cecg}.  We give a proof here to contrast the machinery with the symmetric case in Section~\ref{sec:symextcorr}.

\begin{proposition}[The Revelation Principle for Correlated Equilibria]
\label{prop:corrintext}
The distribution of actions in an indirect correlated equilibrium is a direct correlated equilibrium and every direct correlated equilibrium arises in this way.
\end{proposition}

\begin{proof}
Given an indirect correlated equilibrium, no player can gain by deviating from $f_i$.  In particular no player can gain by deviating to $\zeta_i\circ f_i$ for any $\zeta_i: C_i\to C_i$.  Therefore if we push the $f_i$ back into the unobserved part of the model in Figure~\ref{fig:indirectce}, merging it into the noise-adding stage, the result is an indirect correlated equilibrium in which all players choose the identity function from their information to their action.  This coincides with the definition of what it means for the distribution over information / actions to be a direct correlated equilibrium.

Conversely, given a $\pi\in\Delta(\Gamma)$ we can design a correlation scheme in which the state of the world is a random strategy profile distributed according to $\pi$, each player's information is equal to his personal component of this strategy profile (so the ``noise adding'' step just strips away the other players' choices of strategy -- no additional randomness is needed), and each $f_i$ is the identity on $C_i$.  By the definitions this is an indirect correlated equilibrium if and only if $\pi$ is a direct correlated equilibrium.
\end{proof}

For a given correlation scheme we say that a player \textbf{knows the state of the world} if the state of the world is a function of (measurable with respect to) his information.  This means that player knows all relevant information which he can know in theory, but not the outcomes of the other players' private coin tosses.  We have built some redundancy into the model in the sense that we could have chosen to push all the random noise into the state of the world, making each player's information a deterministic function of the state.  However, the resulting model would not allow for private coin tosses.  Compare the following statements:

\begin{proposition}
\label{prop:knownstatenash}
In an indirect correlated equilibrium, if each player knows the state of the world then the outcome conditioned on the state is a Nash equilibrium almost surely, and every Nash equilibrium arises in this way.
\end{proposition}

\begin{proposition}
In an indirect correlated equilibrium, if each player knows the state of the world and all the noise signals are constant (or equivalently are considered part of the state) then the outcome conditioned on the state is a \emph{pure} Nash equilibrium almost surely, and every pure Nash equilibrium arises in this way.
\end{proposition}

The proofs of these propositions amount to little more than repeating the definitions.  The main idea is that if the players commonly know the state, then conditioned on this knowledge their play is independent and each is best replying to his opponent.  That is to say, the outcome is conditionally a Nash equilibrium, which must be pure if the players cannot privately randomize.

\subsubsection{Fully symmetric indirect correlated equilibria}
\label{sec:symextcorr}
In the introduction we assumed that the players are copies of an identical decision-making agent.  Thus players perform identical experiments and flip identical coins, interpreting and acting on the results in the same way.  Randomness in these outcomes is the one way in which differences in the players' actions arise.  Stated precisely:

\begin{definition}
We say that a correlation scheme is \textbf{fully symmetric}\footnote{We use the term ``fully symmetric'' in place of simply ``symmetric'' to distinguish from arbitrary correlation schemes which result in symmetric distributions over actions.  In particular, we rule out cases in which the state of the world is able to introduce asymmetries between the players.} if all players have the same noise distribution, mapping from state and noise to information, and mapping $f_i$ from information to action.  We say that an indirect correlated equilibrium is \textbf{fully symmetric} if it arises from a fully symmetric correlation scheme.
\end{definition}

The notion of correlation schemes is well-studied \cite{a:scrs}, but to the authors' knowledge this fully symmetric version has not been considered.  Under a fully symmetric correlation scheme, the players' noise distributions are i.i.d., so their information, and hence actions, are i.i.d.\ conditioned on the state.  We immediately obtain the expected extension of Proposition~\ref{prop:knownstatenash}:

\begin{proposition}
\label{prop:knownstatesymnash}
In a fully symmetric indirect correlated equilibrium, if each player knows the state of the world then the outcome conditioned on the state is a symmetric Nash equilibrium almost surely, and every symmetric Nash equilibrium arises in this way.
\end{proposition}

On the other hand the generalization of Proposition~\ref{prop:corrintext} looks a bit different:

\begin{proposition}[The Revelation Principle for Exchangeable Equilibria]
\label{prop:symcorrintext}
The distribution of actions in a fully symmetric indirect correlated equilibrium is an exchangeable equilibrium and every exchangeable equilibrium arises in this way.
\end{proposition}

\begin{proof}
By Proposition~\ref{prop:corrintext} this distribution is a correlated equilibrium.  It is also i.i.d.\ conditioned on the state, so it is exchangeable.

Conversely, let $\pi$ be any exchangeable equilibrium.  Then we can write $\pi = \sum_{i=1}^k \lambda_i x_ix_i^T$ for some $x_i\in\Delta(C_1)$ and some probability vector $\lambda$.  Construct a fully symmetric correlation scheme with $k$ states of the world chosen according to the distribution $\lambda$.  In state $i$, choose each player's information i.i.d.\ according to $x_i$.  Let the mappings from information to action all be the identity map.  Then the distribution over actions is $\pi$, and since this is in particular a direct correlated equilibrium, the resulting correlation scheme is a fully symmetric indirect correlated equilibrium.
\end{proof}

Under sufficiently strong symmetry assumptions, the widely used equivalence between indirect and direct correlated equilibria breaks down.  Instead, fully symmetric indirect correlated equilibria give rise to exchangeable equilibria.

Thus we conclude that there are games (such as Examples~\ref{ex:anticoord},~\ref{ex:exeqsep}, and~\ref{ex:payoffsep} below, or any symmetric game admitting asymmetric Nash equilibria as shown in Theorem~\ref{thm:asymnashstrict}) admitting symmetric direct correlated equilibria which cannot be realized by any fully symmetric indirect correlated equilibria because they fail to be exchangeable.  Implementing such equilibria requires implicit symmetry-breaking.  On the other hand, some direct correlated equilibria outside the convex hull of the Nash equilibria can be implemented without symmetry-breaking; again see Example~\ref{ex:exeqsep}.  Such exchangeable equilibria can also achieve higher payoffs / social welfare than the symmetric Nash equilibria as in Example~\ref{ex:payoffsep}.

Such symmetry-breaking should not be viewed as pathological.  It occurs whenever the state of the world is taken to include information which distinguishes the players, such as their types or some common knowledge about their past interactions.  Such information can break symmetry in play even when the payoffs are symmetric, resulting in non-exchangeable symmetric correlated equilibria or even asymmetric correlated equilibria (the following subsection gives an example).  As such, fully symmetric correlation schemes -- and so exchangeable equilibria -- arise when players correlate their actions based on features of the external world rather than features of the players themselves.

\subsection{Unknown opponent interpretation}
\label{sec:unknownopp}

As in the previous section on the hidden variable interpretation, here we again focus on games in which the symmetry includes not only the action spaces and payoffs, but also the roles of the players in a wider sense.  Kuzmics and Rogers argue formally for symmetric Nash equilibria in symmetric games when players are ignorant about (the payoff-irrelevant types of) their opponents \cite{kr:iijsesg}.  In this section we show by example that a simple interpretation of ignorance of one's opponent leads to exchangeable equilibria when correlation is allowed.

Consider the anti-coordination game of Example~\ref{ex:anticoord} with its strategies named as in Table~\ref{tab:alicebobgame}.  We will first give two examples of choices of players without this symmetry and then an example with this symmetry.  Note that because of the symmetry of the payoffs, we do not have to select who will be the row player and who will be the column player.  In particular this information is not available to help the players coordinate on an equilibrium.
\begin{table}[tbf]
\begin{center}
\begin{tabular}{c|c|c|}
$(u_{\text{Row}},u_{\text{Column}})$ & My name is Alice & My name is Bob \\
\hline
My name is Alice & $(0,0)$ & $(1,1)$ \\
\hline
My name is Bob & $(1,1)$ & $(0,0)$ \\
\hline
\end{tabular}
\end{center}
\caption{Anti-coordinating based on player identity.}
\label{tab:alicebobgame}
\end{table}

Suppose first that this game were played by two people named Alice and Bob whose names were common knowledge.  It is reasonable to predict that the players would have no trouble coordinating on the pure strategy Nash equilibrium in which each player correctly identifies his own name.  Suppose secondly that this game were played by two friends named Alice.  It may still be possible for them to coordinate based on, say, the common knowledge that one of them has the middle name Roberta and the other does not.

Suppose finally that two arbitrary Alices are selected from the record-breaking crowd at this year's Conference of Game Theorists Named Alice.  They are sequestered in separate rooms and asked to choose the action they would play in this game against an opponent chosen from the same conference who has been given the same information.  It is common knowledge that both players are attending the conference, but no further information is given.  As Bayesian observers, how should we expect them to play?  The standard Bayesian way to capture our ignorance of any distinctions between the possible players is to say that the distribution over outcomes should be symmetric.  Based on the conference title we can assume common knowledge of rationality, so by Aumann's result \cite{a:ceebr} we should expect the Alices to play a correlated equilibrium.  In particular we should not expect both to play ``My name is Alice'' with probability one.  

What else can we say about the outcome?  The best symmetric correlated equilibrium in terms of payoff is $W = \left[\begin{smallmatrix}0 & 1/2 \\ 1/2 & 0\end{smallmatrix}\right]$.  Is this a reasonable solution?

We claim that $W$ would not be played.  Suppose for a contradiction that $W$ were the expected distribution of outcomes and consider asking a third Alice -- recall the record attendance -- what action (``Alice'' or ``Bob'') she would play in the game.  The chosen strategies $A_1$, $A_2$, and $A_3$ of the three Alices would all be random variables, and their joint distribution would be symmetric under arbitrary permutations of the Alices.  In particular all the pairwise distributions of two of these random variables would be given by $W$.  Consulting $W$, we see that the events $E_1 = \{A_2 = A_3\}$, $E_2 = \{A_1 = A_3\}$, and $E_3 = \{A_1 = A_2\}$ would each occur with probability zero, therefore so would their union $E = E_1\cup E_2 \cup E_3$.  But there are only two strategies in the game, so in any realization at least two of the $A_i$ must be equal by the pigeonhole principle.  That is to say, $E$ must occur with probability one regardless of the distribution of the $A_i$, a contradiction.

On the other hand, the mixed strategy Nash equilibrium in which each player randomizes over her actions with equal probability does not suffer from this difficulty: one can create a sequence of i.i.d.\ $\text{Bernoulli}(1/2)$ random variables for any number of Alices.  Such a sequence has a distribution which is exchangeable and has the desired mixed equilibrium as its marginal distribution corresponding to any choice of two players.

Strictly speaking we have only argued that the observed correlated equilibrium should extend to an $N$-exchangeable sequence $A_1,\ldots, A_N$, where $N$ is the number of Alices at the conference.  However, with $N$ assumed to be large, it is more natural to consider equilibria which could be extended to an arbitrary number of attendees.  In particular this ensures that we select only the equilibria which are robust to the exact size of the pool and the players' knowledge thereof.  

Being a bit more formal, we consider situations in which the players are drawn from a pool of $N$ interchangeable agents ($2\leq N\leq \infty$).  We interpret interchangeability as meaning that all $N$ agents choose a hypothetical action for the game, regardless of whether they are selected to play, and the joint distribution of these actions is $N$-exchangeable.  Whichever two agents are chosen to play the game, we assume they choose their actions in their best interests, in the sense that their joint distribution is a correlated equilibrium.  That is to say, the $N$ agents choose their actions according to an $N$-exchangeable equilibrium:

\begin{definition}
\label{def:nexcheq}
For $2\leq N\leq \infty$ an \textbf{$N$-exchangeable equilibrium} of $\Gamma$ is an $N$-exchangeable distribution of random variables $X_i$ such that the marginal distribution of $X_1$ and $X_2$ is a correlated equilibrium of $\Gamma$.  The set of $N$-exchangeable equilibria is denoted $\xesym(\Gamma,N)\subseteq\dsym(C_1^N)$.
\end{definition}

The sets of $N$-exchangeable equilibria shrink as $N$ increases, in the sense that if $N_2>N_1$ and $X_1,\ldots,X_{N_2}$ is an $N_2$-exchangeable equilibrium, then $X_1,\ldots,X_{N_1}$ is an $N_1$-exchangeable equilibrium.  Letting $N$ go to infinity, Theorem~\ref{thm:definetti} (De Finetti's Theorem) gives a sort of convergence\footnote{While $N_2$-exchangeable equilibria induce $N_1$-exchangeable equilibria for $N_2>N_1$, we cannot properly say that the sets are nested as ``$\xesym(\Gamma,N_2)\subseteq\xesym(\Gamma,N_1)$'' because these are distributions over different numbers of random variables.  With the language of \textbf{inverse} (or \textbf{projective}) \textbf{limits} it is possible to make sense of the convergence of this sequence of sets as $N$ goes to infinity and show that the limit, which can be thought of as ``$\bigcap_{N=2}^\infty\xesym(\Gamma,N)$'', is $\xesym(\Gamma,\infty)$.  For details see Part~II.} to exchangeable equilibria:

\begin{corollary}
\label{cor:exeqchar}
The distribution of two random variables $X_1,X_2$ with values in $C_1$ is an exchangeable equilibrium if and only if it is the bivariate marginal distribution of an infinite sequence $X_1,X_2,X_3,\ldots$ which is an $\infty$-exchangeable equilibrium.
\end{corollary}

Therefore we may view exchangeable equilibria as exactly those correlated equilibria arising when players are selected from a large pool of interchangeable agents and their identities are unknown to each other.  It makes sense to say that the notion of exchangeable equilibria eliminates those correlated equilibria which are ``unreasonable'' on the grounds of symmetry arguments such as the one made above.  On the other hand, when Alice and Bob play the game in Table~\ref{tab:alicebobgame} these kinds of symmetry arguments do not apply, so it is plausible that they will play an equilibrium which is not exchangeable, perhaps even asymmetric.

\subsection{Many player interpretation}
\label{sec:manyplayerinterp}
Another way to view exchangeable equilibria is as limits of symmetric correlated equilibria of $N$-player games consisting of many identical pairwise interactions\footnote{These games are symmetric versions of polymatrix games \cite{h:epg}.  They are also instances of pairwise-interaction games \cite{dm:pig}, themselves a form of graphical games \cite{kls:gmgt}.}.  Given a symmetric bimatrix game $\Gamma$ we define a symmetric $N$-player game $\npow{\Gamma}$ in which player $i$ has strategy set $\npow{C_i} = C_1$ and utility function
\[
\npow{u_i}(s_1,\ldots,s_N) = \sum_{j\neq i} u_1(s_i,s_j).
\]
That is to say, each player plays $\Gamma$ against each other player, choosing the same strategy in each instance.  If $N=2$ we recover the original game: $\npowk{2}{\Gamma} = \Gamma$.

\savecounter{minority}
\begin{example}
\label{ex:minority}
We give another interpretation of the anti-coordination game $\Gamma$ whose utilities are shown in Table~\ref{tab:alicebobgame}.  The story is different from the one in the previous section, so we simply refer to the strategies as `$A$' and `$B$'.  The utilities of the $N$-player game $\npow{\Gamma}$ are given by
\begin{equation*}
\npow{u}_i(s_1,\ldots,s_N) = \#\{j: s_j\neq s_i\},
\end{equation*}
so each player wants to choose whichever strategy is chosen by fewer of his opponents.  This is a version of ``The Minority Game'' introduced by Challet and Zhang \cite{cz:ecoeg}.  One interpretation is that the $N$ players are choosing between two equally good restaurants, $A$ and $B$, so each player wants to eat at the less crowded restaurant.
\end{example}

The game $\npow{\Gamma}$ is symmetric under arbitrary permutations of the $N$ players, so it is natural to focus on correlated equilibria of this game which are also symmetric, i.e., $N$-exchangeable.  We now show that these are exactly the $N$-exchangeable equilibria of Definition~\ref{def:nexcheq}.  The reason is that the utility for a player in $\npow{\Gamma}$ is summed over all his $N-1$ interactions, so by linearity only the bivariate marginal of an $N$-exchangeable distribution matters to determine whether it is a correlated equilibrium of $\npow{\Gamma}$.  A simple computation then shows that the correct condition on this bivariate marginal is that it be a correlated equilibrium of  $\Gamma$.

\begin{proposition}
\label{prop:nexeqchar}
For finite\footnote{We do not interpret $\infty$-exchangeable equilibria directly as symmetric correlated equilibria of some game to avoid technical difficulties related to games with infinitely many players.} $N$ symmetric correlated equilibria of $\npow{\Gamma}$ are the same as $N$-exchangeable equilibria of $\Gamma$.
\end{proposition}

\begin{proof}
For any $f: C_1\rightarrow C_1$we have
\begin{align*}
\expect \npow{u}_1(f(X_1),X_2,\ldots, X_N) & = \expect\sum_{i=2}^N u_1(f(X_1),X_i) = \sum_{i=2}^N \expect u_1(f(X_1),X_i) \\ 
& = \sum_{i=2}^N\expect u_1(f(X_1),X_2) = (N-1)\expect u_1(f(X_1),X_2),
\end{align*}
with the same analysis holding for all the other players.  Therefore there exists an $f$ which can improve the payoff of some player in $\npow{\Gamma}$ under recommendations $X_1, \ldots, X_N$ if and only if there exists such an $f$ in $\Gamma$ under recommendations $X_1,X_2$.
%
%
%
\end{proof}

Thus exchangeable equilibria correspond to symmetric correlated equilibria of $N$-player extensions of the game for arbitrary $N$.  Given that the original game was symmetric, it is reasonable to expect that equilibria with a high degree of symmetry would be focal (apt to be chosen over other equilibria because they attract attention by being ``better'' in some way).  This is another reason we might expect players to play an exchangeable equilibrium.

\usesavedcounter{minority}
\begin{example}(cont'd)
Which equilibria should we expect to be played in The Minority Game?  There are an abundance of pure equilibria of $\npow{\Gamma}$; these are exactly the strategy profiles in which $\left\lfloor\frac{N}{2}\right\rfloor$ players choose one restaurant and $\left\lceil\frac{N}{2}\right\rceil$ choose the other.  These all have the disadvantage of not being symmetric in the players.  While such equilibria can be justified with an evolutionary model \cite{cz:ecoeg}, they do not make as much sense in a symmetric one-shot game.

We can symmetrize by constructing a correlated equilibrium which picks one of these pure Nash equilibria uniformly at random.  This yields an $N$-exchangeable equilibrium $\pi^N$.  One can show that if $N$ is odd then $\pi^{N+1}$ is the unique extension of $\pi^N$ to an $(N+1)$-exchangeable distribution (the extra player in $\npowk{N+1}{\Gamma}$ goes to the less populated restaurant).  On the other hand, if $N$ is even then there is no $(N+1)$-exchangeable $\pi\in\dsym(C_1^{N+1})$ extending $\pi^N$.  This means $\pi^N$ cannot be extended to an $(N+2)$-exchangeable sequence for any $N$.

In particular, there is no symmetric way to extend any $\pi^N$ to a correlated equilibrium of all the games $\npowk{k}{\Gamma}$ simultaneously.  Therefore expecting the players to actually play one of the $\pi^N$ means assuming that the value of $N\pm 1$ (depending on the parity of $N$) is common knowledge.  This assumption seems dubious in this model.

We claim it is more natural to look for a solution which is robust to the value of $N$.  This could be accomplished in a number of ways, such as by fixing some value $N$ which is much greater than the actual expected number of players and using $\pi^N$ (which only solves the problem in an approximate sense), or by assuming the players have a probability distribution over possible values of $N$ \cite{m:pupg}.  Exchangeable equilibria are those which are the most robust to the value of $N$; namely, they are correlated equilibria of $\npow{\Gamma}$ for all $N$.

The unique exchangeable equilibrium of the anti-coordination game in Table~\ref{tab:alicebobgame} is the mixed Nash equilibrium (proven in Example~\ref{ex:anticoord} below).  That is to say, the only Bayesian rational strategy of the players in The Minority Game which is symmetric and makes sense regardless of the number of players is for everyone to choose a restaurant uniformly at random.
\end{example}
\restorecounter

Returning to an arbitrary symmetric bimatrix game $\Gamma$, the games $\npow{\Gamma}$ are instances of the \textbf{pairwise-interaction games} studied in \cite{dm:pig}.  These are games in which multiple players play many copies of a symmetric bimatrix game, each player choosing the same action in all copies of this base game in which he participates.  The pairwise-interaction structure is specified by a graph whose nodes are players and whose edges specify which pairs of players compete in the base game; in the case of $\npow{\Gamma}$ this is the complete graph on $N$ vertices. The book chapter \cite{dm:pig} covers problems of computing Nash equilibria given this graph as input; here we are concerned with those symmetric correlated equilibria which persist as the graph grows.

The exchangeable equilibria correspond to symmetric correlated equilibria of $\npow{\Gamma}$ simultaneously for all $N$ and this seems to be robust to how exactly we formulate $\npow{\Gamma}$.  In particular an $N$-exchangeable equilibrium of a symmetric bimatrix game is a correlated equilibrium of the associated pairwise-interaction games for all graphs on $N$ vertices, not just complete graphs.  Furthermore, we could allow the players to consider different deviations for different opponents in this interaction graph, in which case the exchangeable equilibria would correspond exactly to the symmetric correlated equilibria in which the players choose the same action against all opponents.  Proofs of both facts mimic that of Proposition~\ref{prop:nexeqchar}.

\subsection{Sealed envelope implementation}
\label{sec:symenvmachine}

For $1\leq N\leq\infty$ we can assign another interpretation to $\xesym(\Gamma,N)$ in terms of which correlated equilibria can be implemented by a certain type of correlation scheme.  Suppose $\Gamma$ is to be played by two players who have access to a common pile of $N$ indistinguishable sealed envelopes, each of which contains a slip of paper on which an element of $C_1$ is written.  The contents of these envelopes play the role of the recommendations in the formulation of (direct) correlated equilibrium.

Each player is allowed to choose one envelope and base his choice of action in $\Gamma$ on its contents, which he examines privately.  We assume that there is some mechanism in place preventing both players from choosing the same envelope (e.g. a random one of them is required to select a different envelope).

The fact that the envelopes are indistinguishable means that the distribution over their contents should be exchangeable, since we would lose no information by shuffling the envelopes randomly.  This distribution is an $n$-exchangeable equilibrium if and only if it is a Nash equilibrium of this larger game (including the envelope selection process) for both players to choose distinct envelopes in an arbitrary fashion and then play the strategy written inside the envelope they choose.  In particular, they cannot exploit the knowledge of which envelope the other player will choose in any way.

\section{Examples}
\label{sec:xeexamples}
In this section we compare the sets of Nash, exchangeable, and correlated equilibria of some small example games.  The main tool for doing so is Theorem~\ref{thm:cpvsdnn}.

Our first two examples are $2\times 2$ games, including the anti-coordination game discussed in the introduction, illustrating the general fact that $\xesym = \conv(\nesym)$ for such games (Theorem~\ref{thm:2by2}).  Third we show that the sets $\conv(\nesym)$, $\xesym$, and $\cesym$ can all be distinct for $3\times 3$ games and that $\xesym$ need not be polyhedral.  Fourth and finally we give a $3\times 3$ example showing the the maximum values of the expected utility / social welfare achievable by symmetric correlated, exchangeable, and Nash equilibria can all be distinct.

All examples except the fourth are symmetric bimatrix games of identical interest, i.e., they satisfy $A = A^T = B$.  The richness of behavior that arises in these games is interesting because identical interest games ($A = B$) are often thought of as somewhat trivial; any payoff maximizer is a pure Nash equilibrium, for example.   However, the pure Nash equilibria need not be symmetric, so when restricting attention to symmetric equilibria identical interest games seem to be a less trivial class.  

\begin{example}
\label{ex:coord}
Consider the coordination game with utility matrices $A = B = \left[\begin{smallmatrix}1 & 0 \\ 0 & 1\end{smallmatrix}\right]$.  Parametrize symmetric probability matrices as $\left[\begin{smallmatrix} p & q \\ q & r\end{smallmatrix}\right]$ with $p,q,r\geq 0$ and $p+2q+r =1$. Then the correlated equilibrium conditions are exactly $p\geq q$ and $r\geq q$.  Note that together these imply that $pr\geq q^2$, so any symmetric correlated equilibrium is automatically positive semidefinite (the principal minors are nonnegative), hence conditionally i.i.d.\ by  Theorem~\ref{thm:cpvsdnn} and so exchangeable.  Therefore the set of symmetric exchangeable equilibria is the same as the set of symmetric correlated equilibria and a simple computation shows that these sets are equal to
\begin{equation*}
\cesym = \xesym = \conv\left\{\begin{bmatrix}1 & 0 \\ 0 & 0\end{bmatrix}, \begin{bmatrix}0 & 0 \\ 0 & 1\end{bmatrix},\begin{bmatrix} 1/4 & 1/4 \\ 1/4 & 1/4 \end{bmatrix}\right\} = \conv(\nesym).
\end{equation*}
\end{example}

\begin{example}
\label{ex:anticoord}
Consider the anticoordination game with utility matrices $A = B = \left[\begin{smallmatrix} 0 & 1 \\ 1 & 0\end{smallmatrix}\right]$ discussed in Section~\ref{sec:unknownopp} and again in the context of The Minority Game in Section~\ref{sec:manyplayerinterp}.  Up to a positive affine transformation of the utilities this game is the same as the game Chicken analyzed in Section~\ref{sec:preview}, and so has the same equilibria:
\begin{equation*}
\begin{split}
\left\{\begin{bmatrix}1/4 & 1/4 \\ 1/4 & 1/4\end{bmatrix}\right\} & = \nesym = \xesym \\ & \subsetneq \cesym = \conv\left\{\begin{bmatrix} 0 & 1/2 \\ 1/2 & 0\end{bmatrix}, \begin{bmatrix} 1/3 & 1/3 \\ 1/3 & 0\end{bmatrix}, \begin{bmatrix} 0 & 1/3 \\ 1/3 & 1/3\end{bmatrix}, \begin{bmatrix} 1/4 & 1/4 \\ 1/4 & 1/4\end{bmatrix}\right\}.
\end{split}
\end{equation*}
\end{example}

The games in Examples~\ref{ex:coord} and~\ref{ex:anticoord} could be said to be isomorphic as bimatrix games but not as symmetric bimatrix games, in the sense that one can be transformed into the other by relabeling the strategies, but not by using the same relabeling for both players.  This accounts for the fact that their sets of symmetric exchangeable equilibria are not isomorphic.  The difference corresponds to the fact that when these games are played by many players simultaneously, there is a symmetric (with respect to the players) way to break symmetry between the two pure strategies in the case of the coordination game but not in the case of the anticoordination game.  

Also note that in both cases we had $\xesym = \conv(\nesym)$ as predicted by Theorem~\ref{thm:2by2}.  We now construct a $3\times 3$ game with $\conv(\nesym)\subsetneq \xesym\subsetneq \cesym$.

\savecounter{exeqsep}
\begin{example}
\label{ex:exeqsep}
Consider the game with utilities
\[
A = B = \begin{bmatrix}2 & 2 & 0 \\ 2 & 1 & 2 \\ 0 & 2 & 2\end{bmatrix}.
\]
First, we show strict containment in $\xesym\subsetneq
\cesym$.  Second, we show strict containment in $\conv(\nesym)\subsetneq\xesym$.  Third, we observe that $\xesym$ is not polyhedral.  These results are summarized in Figure~\ref{fig:exeqsep}.


\begin{figure}[tbp]
\begin{center}
\includegraphics{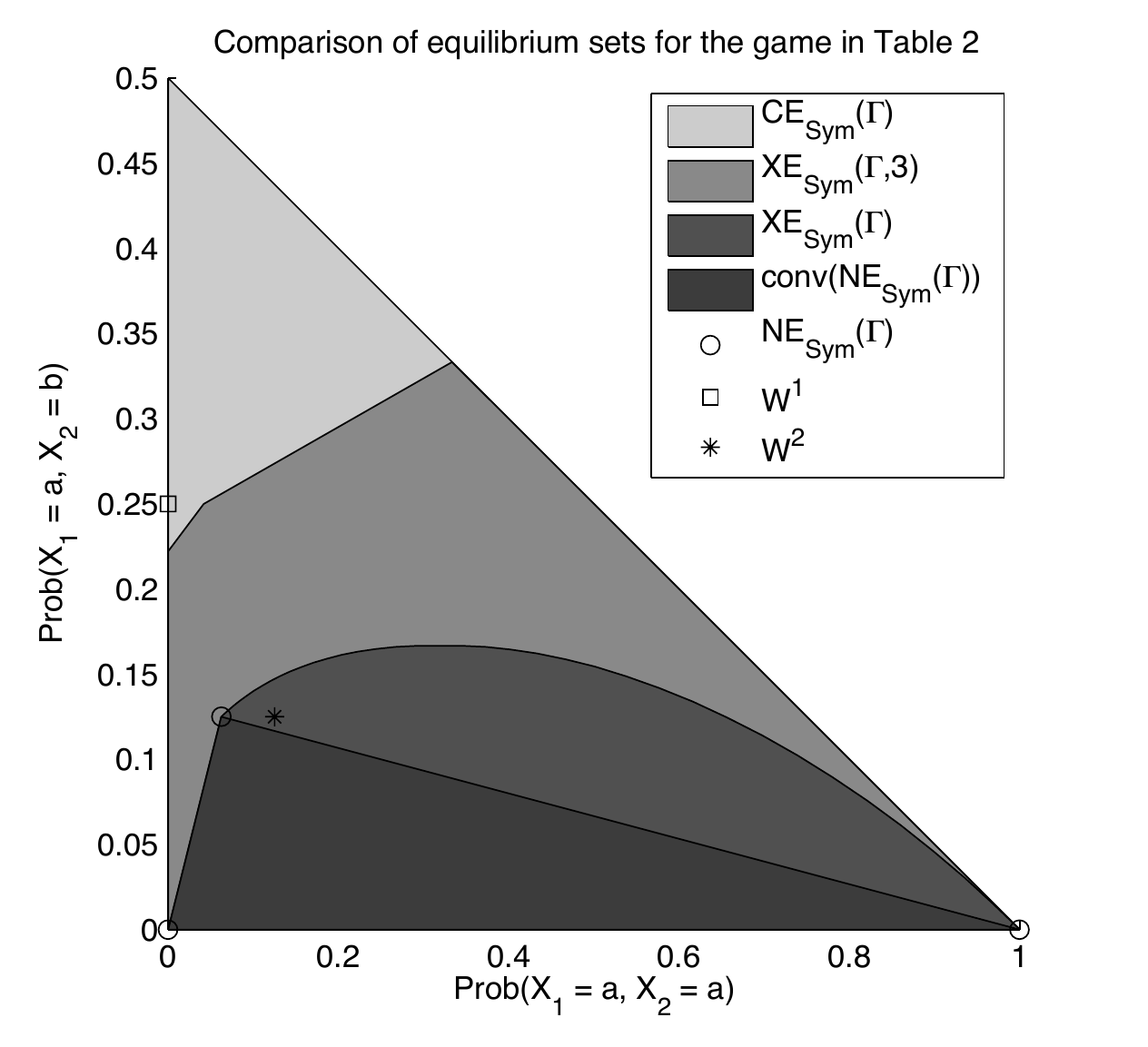}
\caption{Comparison of equilibrium sets for the game in Example~\ref{ex:exeqsep}.  These sets are naturally sets of symmetric $3\times 3$ matrices (or in the case of $\xesym(\Gamma,3)$, a set of $3\times 3\times 3$ arrays), but we have chosen a projection into two dimensions which highlights the separation between the sets.  The set of exchangeable equilibria is not polyhedral despite the fact that the other sets are.}
\label{fig:exeqsep}
\end{center}
\end{figure}

The matrix
\[
W^1 = \frac{1}{4}\begin{bmatrix}0 & 1 & 0 \\ 1 & 0 & 1 \\ 0 & 1 & 0\end{bmatrix}
\]
is a correlated equilibrium because both players get their maximum payoff with probability one.  By Proposition \ref{prop:cpzeros} it is not conditionally i.i.d., so it is a correlated equilibrium which is not exchangeable.  In fact we can say more: $W^1$ cannot be extended to a $3$-exchangeable distribution.  Suppose $X_1,X_2,X_3$ were random variables with a $3$-exchangeable distribution  taking values in $\{a,b,c\}$ such that the distribution of $X_1$ and $X_2$ were $W^1$.  Then
\begin{align*}
\frac{1}{4}=\prob(X_1 = a) &  \leq \prob\left((X_1 = a,X_2\neq b) \vee (X_1 = a,X_3\neq b) \vee (X_2 = b,X_3=b)\right) \\
& \leq \prob(X_1=a,X_2\neq b) + \prob(X_1=a,X_3\neq b) + \prob(X_2=b,X_3=b) \\
& = 2W^1_{11}+2W^1_{13} + W^1_{22} = 0,
\end{align*}
which is a contradiction, so no such $3$-exchangeable distribution exists.

One can also verify that
\[
W^2 
= \frac{1}{8}\begin{bmatrix}1 & 1 & 0 \\ 1 & 2 & 1 \\ 0 & 1 & 1\end{bmatrix}
= \frac{1}{8}\begin{bmatrix}1 & 0 \\ 1 & 1 \\ 0 & 1\end{bmatrix}\begin{bmatrix}1 & 1 & 0 \\ 0 & 1 & 1\end{bmatrix}
\]
is a correlated equilibrium, and the exhibited factorization shows that $W^2$ is conditionally i.i.d., hence exchangeable.  This factorization can be interpreted as a fully symmetric correlation scheme: the state of the world is $s = a$ or $s=c$ with equal probability and conditioned on the state, the players are each independently told to play either $s$ or $b$ with equal probability.  Or equivalently, the players make independent measurements of the state which succeed and return the true state $s$ half the time and fail and return ``inconclusive'' the other half of the time, leading the players to play $s$ and $b$, respectively.

Given this information structure it is an equilibrium for each player to play his recommendation / measurement.  Note that if the players always knew the true state and were never told `$b$' then this would no longer be an equilibrium.  That is to say, the distribution of actions conditioned on the state is not always a Nash equilibrium; indeed it is not in either state.

In fact $W^2$ is not a convex combination of Nash equilibria at all.  Suppose for a contradiction that it were. Then at least one of the Nash equilibria in the convex combination would have to assign positive probability to the strategy profile $(b,b)$.  Suppose player $2$ did not play $c$ with positive probability in such a Nash equilibrium.  Given this information, player $1$ prefers $a$ to $b$, so player $1$ cannot choose $b$ with positive probability in such an equilibrium, a contradiction.  Symmetric arguments show that each player must play all his strategies with positive probability.  Therefore this Nash equilibrium has full support.  But $W^2$ has entries which are zero, hence this Nash equilibrium cannot be included in an expression of $W^2$ as a convex combination of Nash equilibria.  Thus $W^2$ is not a convex combination of Nash equilibria.

This argument also shows that the only symmetric Nash equilibrium which assigns positive probability to $b$ is the one with full support, which a simple computation shows to be $\begin{bmatrix}\frac{1}{4} & \frac{1}{2} & \frac{1}{4}\end{bmatrix}$.  The only other symmetric Nash equilibria are $\begin{bmatrix}1 & 0 & 0\end{bmatrix}$ and $\begin{bmatrix}0 & 0 & 1\end{bmatrix}$.  There can be no symmetric Nash equilibrium which assigns zero probability to $b$ but positive probability to both $a$ and $c$.  Such an equilibrium would have to assign equal probability to $a$ and $c$, but $b$ is the unique best response to such a mixture.

As is well known, the set of correlated equilibria of a game is always polyhedral, hence so is the set of $k$-exchangeable equilibria.  The set of Nash equilibria is generically finite, so its convex hull is polyhedral for generic games (and for this game in particular).  It is visually evident, and can be proven algebraically, that the projection of the set of exchangeable equilibria pictured in Figure~\ref{fig:exeqsep} is not polyhedral.  In fact it is an algebraic curve of degree $11$ which factors into three linear components, a quadratic component, and a degree six component over $\Qm$.  Two of the linear components are easily visible (the bottom and left of the convex hull of the symmetric Nash equilibria), and the third corresponds to the maximum $y$ value, attained along the line segment joining $(\frac{11}{36},\frac{1}{6})$ to $(\frac{1}{3},\frac{1}{6})$.  The quadratic component corresponds to the curved portion of the boundary to the right of this maximum and is defined by the vanishing of $x^2 + 2xy + 4y^2 - x$.  The degree six component is the curved portion of the boundary to the left of the maximum; we omit its equation for brevity.
\end{example}

\begin{example}
\label{ex:payoffsep}
We show that the maximum expected utilities achievable by symmetric correlated equilibria, exchangeable equilibria, and symmetric Nash equilibria are all distinct in the game $\Gamma$ with payoffs
\begin{equation*}
A = B^T = \begin{bmatrix} 0 & 1 & 1 \\  2 & 1 & 0 \\ 1 & 0 & 1\end{bmatrix}.
\end{equation*}
The same is true with social welfare in place of utility; for symmetric bimatrix games and symmetric probability distributions both players receive the same expected utility, so the social welfare is just twice this utility.  The symmetric probability distribution
\begin{equation*}
W^1 = \begin{bmatrix} 0 & 1/2 & 0 \\ 1/2 & 0 & 0 \\ 0 & 0 & 0\end{bmatrix}
\end{equation*}
gives both players an expected utility of $\frac{3}{2}$.  Symmetric distributions always give both players the same utility in symmetric bimatrix games and $W^1$ is the only symmetric distribution (equilibrium or not) which yields a utility this high.  One can verify that $W^1$ is a correlated equilibrium, but it is not exchangeable by Proposition~\ref{prop:cpzeros} because it has zero diagonal, so $W^1\in\cesym(\Gamma)\setminus\xesym(\Gamma)$.  Furthermore $W^1$ achieves a higher utility than any exchangeable equilibrium.

The distribution
\begin{equation*}
W^2 \eqdef \frac{5}{7}\begin{bmatrix} 1/8 \\ 7/8 \\ 0\end{bmatrix}\begin{bmatrix} 1/8 \\ 7/8 \\ 0\end{bmatrix}^T + \frac{2}{7}\begin{bmatrix} 1/8 \\ 0 \\ 7/8 \end{bmatrix}\begin{bmatrix} 1/8 \\ 0 \\ 7/8\end{bmatrix}^T = \frac{1}{64}\begin{bmatrix} 1 & 5 & 2 \\ 5 & 35 & 0 \\ 2 & 0 & 14\end{bmatrix}
\end{equation*}
is exchangeable and it is also a correlated equilibrium, so it is an exchangeable equilibrium yielding expected utility $\frac{17}{16}$ to both players.

Suppose for a contradiction that there were a symmetric Nash equilibrium $(x,x)$ with expected utility $x^T A x > 1$.  Since $x$ is a probability vector, $[Ax]_1,[Ax]_3\leq 1$.  But $x^T A x$ is a convex combination of $[Ax]_1$, $[Ax]_2$, and $[Ax]_3$, so $[Ax]_2>1$.  Therefore strategy $2$ is the unique best response to $x$.  Only best responses are played with positive probability in a Nash equilibrium, so $x = \left[\begin{smallmatrix} 0 & 1 & 0\end{smallmatrix}\right]$.  But then the expected utility is $x^T A x = 1$, a contradiction.  Thus there is no symmetric Nash equilibrium which yields utility greater than $1$.  In particular $W^2\in\xesym(\Gamma)\setminus\conv(\nesym(\Gamma))$ and $W^2$ achieves a higher utility than any symmetric Nash equilibrium.

Indeed, Gambit \cite{mmt:gambit} computes that there are two symmetric pure Nash equilibria $\left[\begin{smallmatrix} 0 & 1 & 0\end{smallmatrix}\right]$ and $\left[\begin{smallmatrix} 0 & 0 & 1\end{smallmatrix}\right]$ as well as one symmetric mixed Nash equilibrium $\left[\begin{smallmatrix} \frac{1}{4} & \frac{1}{4} & \frac{1}{2}\end{smallmatrix}\right]$.  Both pure equilibria yield utility $1$ and the mixed equilibrium yields utility $\frac{3}{4}$.
\end{example}

\section{Conclusions and Future Work}
\label{sec:conclusions}
In the case of symmetric bimatrix games, we have argued that the exchangeable equilibria are a natural mathematical object with a variety of game-theoretic interpretations.  The set of exchangeable equilibria lies between the sets of Nash and correlated equilibria, and its structural properties are a mix of the two.  It is convex like the correlated equilibria, but its semialgebraicity and nonpolyhedrality remind one more of Nash equilibria.

The exchangeable equilibria are not the only object living in the gap between the Nash and correlated equilibria.  There are also, for example, Sorin's distribution equilibria \cite{s:dede}, which neither contain nor are contained in the exchangeable equilibria.  We believe there is room for more insight to be gained interpolating between Nash and correlated equilibria.

In this paper we have largely ignored the question of what ``exchangeable equilibrium'' should mean in symmetric games with more than two players.  We extend the theory to such games in Part~II, studying what is required of the symmetry structure to obtain theoretical results like those of Section~\ref{sec:xebasic} and interpretations along the lines of Section~\ref{sec:interp} of this paper.  This more abstract setting is a natural one for proving existence of exchangeable equilibria by adapting Hart and Schmeidler's argument \cite{hs:ece}.  We also modify the sealed envelope interpretation of Section~\ref{sec:symenvmachine} to design a hierarchy of nested convex sets called higher order exchangeable equilibria interpolating between the exchangeable equilibria and the convex hull of the symmetric Nash equilibria, to which they converge in the limit.


These ideas and relevant computational questions are all explored further in the first author's doctoral thesis \cite{s:doctoralthesis}.

\section*{Acknowledgements}
The authors would like to acknowledge Muhamet Yildiz and Dirk Bergemann for their helpful insights and suggestions on this topic.  This research was funded by the National Science Foundation under Award 1027922 and the Air Force Office of Scientific Research under grant FA9550-11-1-0305.


\bibliographystyle{plain}
\bibliography{../../references}
\end{document}